\documentclass{scrartcl}

\usepackage[T1]{fontenc}
\usepackage[utf8]{inputenc}

\usepackage{amsmath, amssymb, amsfonts, amsthm}
\usepackage{mathtools}
\usepackage{dsfont} % indicator function \mathds{1}
\makeatother
\usepackage{graphicx}
\usepackage[authoryear]{natbib}
\usepackage{hyperref}
\usepackage[capitalize]{cleveref}

\newcommand*{\E}{\mathbb{E}}  % expectation
\newcommand*{\R}{\mathbb{R}}  % set of real numbers
\newcommand*{\N}{\mathbb{N}}  % set of (non-negative) integers
\newcommand*{\pr}{\mathbb{P}}  % probability of event
\newcommand*{\var}{\mathrm{var}} % variance
\newcommand*\dd{\mathop{}\!\mathrm{d}} % d in dx in an integral
\newcommand*{\ind}{\mathds{1}}  % indicator function
\newcommand*{\Unif}{\mathcal{U}}  % uniform distribution
\newcommand*{\OO}{\mathcal{O}}  % big O
\newcommand*{\fhat}{\widehat{f}_N}  % control variate
\newcommand*{\setC}[1]{\mathcal{C}^{#1}\left([0, 1]^s\right)}
\newcommand*{\setCr}{\setC{r}}
\newcommand*{\Ih}{\widehat{I}}
\newcommand*{\Iest}[1]{\Ih_\mathrm{#1}}

\DeclareMathOperator{\average}{Average}
\DeclareMathOperator{\empvar}{EmpVar}
\DeclareMathOperator{\rmse}{RMSE}
\DeclareMathOperator{\vol}{Vol}

\newtheorem{prop}{Proposition}
\newtheorem{thm}{Theorem}
\newtheorem{exa}{Example}
\newtheorem{lem}{Lemma}

\title{Adaptive stratified Monte Carlo using decision trees}
\author{Nicolas Chopin (ENSAE, IPP) 
  \and Hejin Wang (Tsinghua University) 
  \and Mathieu Gerber (University of Bristol)}

\date{}

\begin{document}

\maketitle

\begin{abstract}
  It has been known for a long time that stratification is one possible
  strategy to obtain higher convergence rates for the Monte Carlo estimation of
  integrals  over the hyper-cube $[0, 1]^s$ of dimension $s$. However,
  stratified estimators such as Haber's are not practical as $s$ grows, as they
  require $\OO(k^s)$ evaluations for some $k\geq 2$. We propose an adaptive
  stratification strategy, where  the strata are derived from a decision tree
  applied to a preliminary sample. We show that this strategy leads to higher
  convergence rates, that is the corresponding estimators converge at rate
  $\OO(N^{-1/2-r})$ for some $r>0$ for certain classes of functions.
  Empirically, we show through numerical experiments that the method may
  improve on standard Monte Carlo even when $s$ is large.
\end{abstract}

\section{Introduction}\label{sec:Introduction}

\subsection{General motivation}

This paper is concerned with the unbiased estimation of intractable integrals
with respect to the unit hyper cube of dimension $s\geq 1$:
\begin{equation*}
	I(f) = \int_{[0, 1]^s} f(x) \dd x.
\end{equation*}
A well-known approach to this problem is the Monte Carlo estimator:
\begin{equation*}
	\Iest{MC} = \frac{1}{N} \sum_{n=1}^N f(X_n)
\end{equation*}
which relies on $N$ independent and identically distributed variables,
$X_n\sim \Unif([0, 1]^s)$. The RMSE (root mean square error) of
$\hat{I}_{\text{MC}}$ is $\OO(N^{-1/2})$. We wish to derive estimators with a
faster convergence rate (at least for certain classes of functions).

Before moving on, we briefly recall the advantages of unbiased estimators of
integrals over other (deterministic or stochastic) approximations of integrals.
First, one may evaluate several realisations of an unbiased estimator (possibly
in parallel) and compute (a) their average to obtain a more accurate estimator;
and (b) their empirical variance to assess the numerical error.    Second,  several
numerical schemes in optimisation and sampling may remain valid when an
intractable quantity is replaced by an unbiased approximation. This is the case
for instance in pseudo-marginal samplers \citep{Andrieu2009}, which are Markov
chain Monte Carlo samplers where an intractable target density (or a similar
quantity) is replaced by an unbiased estimate; or in stochastic approximation
algorithms \citep{MR42668} which are gradient-based root-finding algorithms,
where the gradient is also replaced by an unbiased estimate.

\subsection{Control variates, optimal rates, and complexity}\label{sub:cv}

A well-known approach to derive lower-variance estimators is to use control
variates, that is:
\begin{equation}\label{eq:CVest}
	\hat{I}_{\text{CV}} = \frac{1}{N} \sum_{n=1}^N \left\{f(X_n) - \fhat(X_n)\right\} + I(\fhat)
\end{equation}
for a function $\fhat$ whose integral may be computed exactly. Since
\begin{equation}\label{eq:var_CVest}
	\var\left[ \hat{I}_{\text{CV}} \right] \leq \frac 1 N \|f-\fhat\|_\infty^2
\end{equation}
one may obtain a higher convergence rate by setting $\fhat$ to some
approximation of $f$, based on $N$ preliminary evaluations of function $f$, in such a way
that $\|f-\fhat\|_\infty = \OO(N^{-\alpha})$ for some $\alpha>0$.
Note that this approach requires therefore $2N$ evaluations of $f$.

This idea actually goes back to the very beginning of the history of the Monte
Carlo method. In particular,  \citet{MR115275} established that (loosely
speaking) the optimal convergence rate of any estimator based on $N$
evaluations of $f$ is $\OO(N^{-1/2-r/s})$ for functions $f\in\setCr$, the set
of $r-$times continuously differentiable functions.  His proof relies
explicitly on~\eqref{eq:CVest}, and the fact that the optimal rate for function
approximations is $N^{-r/s}$ on $\setCr$. See \citet{novak2016some} for a more
precise statement and a discussion of \citet{MR115275}' results.

More recently, \eqref{eq:CVest} have received renewed interest as a way of
deriving practical estimators of integrals, using also a preliminary sample of
size $N$, and various non-parametric estimation techniques to obtain $\fhat$.
\citet{controlfuncs} use a Bayesian approach based on a Gaussian process prior.
The main drawback of their approach is that computing $\fhat$ has $\OO(N^3)$
complexity. \citet{LelucPortierSegersZhuman} uses instead a $k-$nearest
neighbour estimator for $\fhat$; but computing $\fhat(X_n)$ for each $n$
typically has complexity $\OO(N^2)$.  These approaches remain appealing when
$f$ is expensive to compute. In that case, the $2N$ evaluations of $f$ may
remain more expensive than the computation of $\fhat$ for practical values of
$N$, even if the latter part has a larger complexity. Still, one may want to
derive estimators that have a linear, or close-to-linear complexity in $N$.
This is one of the objectives of this paper.

We mention one last possible control variate strategy: assume $N=k^s$, for some
integer $k\geq 2$, and split the hyper-cube $[0, 1]^s$ into $N$ hyper-cubes of
edge length $1/k$ (and hence volume $1/N$). Let $C_n$, $n=1, \dots, N$, denote
these $N$ `sub-cubes', let $c_n$ be the centre of $C_n$, and define $\fhat$ as
\begin{equation*}
	\fhat(x) = \sum_{n=1}^N f(c_n) \ind_{C_n}(x).
\end{equation*}
See \cref{fig:strat} for a pictorial representation.

\begin{figure}
	\begin{center}
		\includegraphics[scale=0.4]{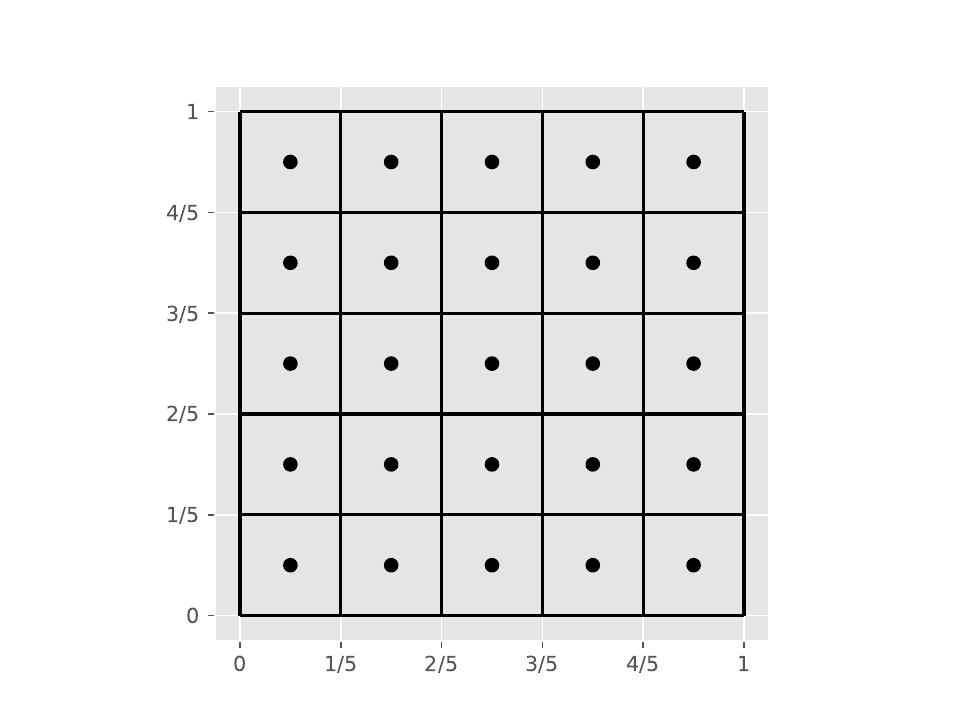}
	\end{center}
	\caption{Cubic stratification in dimension $s=5$, with $k=5$, $N=25$.
		The black dots represent the centers $c_n$.}
	\label{fig:strat}
\end{figure}

If $f\in\setC{1}$, then $\|f-\widehat{f}_N\|_\infty=\OO(k^{-1})=\OO(N^{1/s})$,
and, by~\eqref{eq:var_CVest}, the corresponding estimator is optimal for $r=1$.
It turns out that one can obtain the same rate with a simpler estimator, as
explained in the next section. Still, it will be useful in what comes next to
keep this particular control variate in mind.

\subsection{Stratification and optimal rates}

Stratification is another well-known variance reduction strategy in Monte
Carlo, see, e.g. Chap.~4 of \citet{Lemieux:MCandQMCSampling} or Chap.~8 of
\citet{Owen2018}. It amounts to splitting the integration domain into $P$ strata,
and (assuming that the strata have the same volume, and that $P$ divides $N$),
to generate in each stratum $N/P$ points.

Consider in particular as strata the $N$ sub-cubes mentioned in the previous
section; that is, consider the following estimator:
\begin{equation*}
	\Iest{Haber} = \frac 1 N \sum_{n=1}^N f(X_n),
	\qquad X_n\sim \Unif\left(C_n\right).
\end{equation*}
This estimator has been proposed by \citet{haber1966modified} and is commonly
referred to as the Haber estimator of order one.
Note the similarity with the control variate estimator mentioned at the end of
the previous section. It is easy to check that it has the same convergence
rate, $\OO(N^{-1/2-1/s})$, for $f\in\setCr$, and is therefore optimal as well (for
$r=1$). In practice, Haber's estimator is cheaper to compute (as it relies on
$N$ evaluations, rather than $2N$), and has lower variance (as can be shown by
a simple Rao-Blackwell argument). More generally, it is worth noting that
control variates and stratification are very close in spirit, and especially so
when the considered control variate is piecewise constant.

Haber's estimator has been generalised to $r=2$ by \cite{haber1967modified},
and to $r\geq 3$ by \cite{strat}.
The main drawback of these estimators  is that they are defined only for $N=k^s$,
$k=2,\ldots$. This is impractical whenever $s\gg 10$.
Even when $s \leq 10$, dividing the domain into equal strata may be sub-optimal
whenever $f$ vary more in certain directions than in others. Consider the
extreme case where $s=2$, and $f$ is constant in the first component of $x$,
but in the second. Then it would make more sense to slice the domain $[0, 1]^2$
into $N$ horizontal slices, i.e., $[0, 1]\times [(n-1)/N, n/N]$ for $n=1,\dots,
	N$.

\subsection{Proposed approach}

We develop in this paper an adaptive stratified Monte Carlo approach, where the
strata are automatically adapted to the integrand $f$. To that effect, we
generate a preliminary sample of size $N$ (as in control variates), and we
train a decision (regression) tree on this data to construct the strata.

Decision trees are particularly convenient in our
context, for two reasons, one fundamental, and one computational. The
fundamental reason is that the regression function derived from a decision tree
is piecewise constant, with `pieces' constructed recursively in order to
minimise the variance of the estimator. We use these `pieces' as strata.

The practical one is that it learning and evaluating  regression trees have
complexity $\OO(N \log N)$. Thus, the overall complexity of our approach will
also  be $\OO(N \log N)$, which is more appealing than the polynomial
complexity of approaches based on non-parametric control variates
\citep[again, see][]{controlfuncs, LelucPortierSegersZhuman}.

Finally, our approach will work for any sample size $N=2^k$, and could be
extended easily to any arbitrary value of $N$, as we shall explain.

\subsection{Plan and notations}

\Cref{sec:algo} presents decision trees, and describes the
proposed estimator.
\Cref{sec:convergence} develops some supporting theory for the proposed
estimator.
\Cref{sec:numerics} presents a numerical study to determine how
practical and efficient is the proposed estimator.
\Cref{sec:conclusion} discusses future research. 
An appendix contains proofs of the main results.

We use the short-hand $1:N$ for the set of integers $\{n:\,1\leq n \leq N\}$.
For $x\in\R^s$, we write its $j$-th component as $x[j]$. For any finite set
$\mathcal{A}$ of points in $\R^s$, of size $|\mathcal{A}|$, we denote by
$\average(\mathcal{A})$ their average (empirical mean), assuming
$|\mathcal{A}|\geq 1$, and by $\empvar(\mathcal{A})$ their empirical variance:
\[\empvar(\mathcal{A}) =
	\frac{1}{|\mathcal{A}| - 1} \sum_{x\in\mathcal{A}} \{x - \average(\mathcal{A})\}^2
\]
assuming $|\mathcal{A}|\geq 2$. In case $|\mathcal{A}| \leq 1$, we abuse
notations and  set $\empvar(\mathcal{A}) = + \infty$. This non-standard
convention will be convenient in the specific context of this paper.

We denote by $\rmse\left(\Iest{}\right)$ the RMSE of an estimator $\Iest{}$.
Since we deal only with unbiased estimator, this quantity is always the square
root of $\var(\widehat{I})$, the variance of $\widehat{I}$.

\section{Proposed algorithm}\label{sec:algo}

\subsection{Decision trees}\label{sub:tree}

Decisions trees are classical supervised learning methods, which go back to
\citet{decision_tree_paper}. We focus on regression decision trees, which, given
data $(X_n, Y_n)_{n=1,\dots,N}$, taking values in $[0, 1]^s\times \R$,
constructs a piecewise constant predictor:
\begin{equation}\label{eq:def_decision_tree}
	\widehat{f}_N(x)
	= \sum_{p=1}^P y_p \times \ind\{ x \in R_p\},
	\qquad y_p = \average\{Y_n:\, X_n\in R_p\}
\end{equation}
based on a partition of the feature space, $[0, 1]^s$, which corresponds to a
collection $\mathcal{R}=(R_p)_{p=1,\dots, P}$ of rectangles (multi-variate
intervals).

\begin{figure}
\begin{center}
  \includegraphics[scale=0.4]{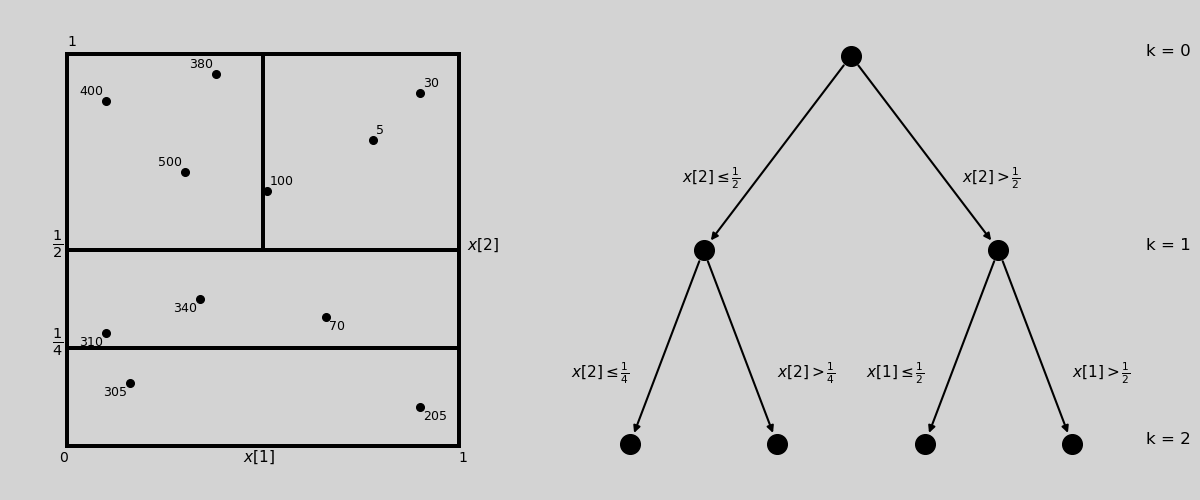}
\end{center}
\caption{Partition of $[0, 1]^2$ into rectangles (left) that correspond to a
decision tree (right).}
\label{fig:tree}
\end{figure}

This collection is constructed by growing greedily (recursively) a tree, where
each node corresponds to a simple rule, $X[j] \leq a$ or $X[j] > a$, and the
leaves of the tree defines the rectangles $R_p$; see \cref{fig:tree}. Given a
rectangle $R$, let $R[j]^+$ and $R[j]^-$ define the two (equal size)
sub-rectangles  obtained by splitting $R$ in the middle, along axis $j$. That
is, if $R=\prod_{i=1}^s [a_i,b_i]$, then $R[j]^-=\prod_{i=1}^s [a_i',b_i']$,
with $(a_i',b_i')=\left(a_i', b_i'+(b_i'-a_i')/2\right)$ if $i=j$,
$(a_i',b_i')= (a_i, b_i)$ otherwise; and $R[j]^+$ is defined similarly.

We initialise $\mathcal{R}$ as follows: $\mathcal{R}\gets \{ [0, 1]^s \}$.
Then, recursively, we  replace each rectangle $R$ in $\mathcal{R}$ by
$R[j_\star]^+$ and $R[j_\star]^-$,  where
$j_\star = \arg\min_{j\in 1:s} \widehat{\Delta}(R, j)$, and
$\widehat{\Delta}(R, j)$ is the so-called CART (classification and regression
tree) criterion:
\begin{equation}\label{eq:criterion_tree}
	\widehat{\Delta}(R, j) :=
	\frac 1 2 \empvar\{Y_j: X_j\in R[j]^+\}  + \frac 1 2  \empvar\{Y_j: X_j\in
	R[j]^-\}.
\end{equation}
In words, we are trying to construct recursively a partition such that the
$Y_i$ have as little variability as possible within each stratum. 
% Recall that $\widehat{\Delta}(C, j) = +\infty$ whenever $R[j]^+$ or $R[j]^-$
% contains strictly less that two points (as per the definition of
% $\empvar(\mathcal{A})$ for $\mathcal{A}$ such that $|\mathcal{A}|\leq 1$).

We stop the splitting when the have reached a (pre-specified) depth $d$ (in
that case, the volume of the considered rectangle is $2^{-d}$).
In case of equality (more than one component $j$ achieves the smallest CART
criterion), we choose randomly. This happens in particular when a given
rectangle contains less than two points; then $\widehat{\Delta}(R, j)=\infty$
for all $j$, and again we choose randomly $j_\star$ among $1, \dots, s$.

We have described only one particular, simple way to learn a decision tree.
Variants may be obtained by considering a criterion different
from~\eqref{eq:criterion_tree}; by allowing to split at an arbitrary location
(rather than in the middle); or by introducing pruning procedures (where the
tree is grown until it reaches a certain depth, then some of its branches are
pruned according to another criterion). For more details on all these variants,
we refer to Chap.~9 of \cite{MR2722294}.  The particular version we have
focused on this section happens to fit well with our approach, as we explain in
next section.

\subsection{Proposed estimator}\label{sub:algo}

We now describe our approach. We assume $N=2^k$, for a certain integer $k\geq
	1$. We proceed in two steps.

In the first step, we generate $N$ independent random variates 
$X_n\sim \Unif[0, 1]^s$, we let $Y_n = f(X_n)$, and we use the data 
$(X_n, Y_n)_{n\in[N]}$ to learn a decision tree, as explained in the previous section.
We obtain a collection of rectangles $(R_p)_{p\in 1:P}$. By construction, the
volume of $R_p$ is $2^{-d}$ where $d$ is the depth of the tree, and $P=2^d$. 

In a second step, we generate independently, for each $p$, and for $n=1,\dots,
	N_p := N 2^{-d}$, variate $\tilde{X}_{n,p} \sim \Unif(R_p)$. The
proposed estimator (called AdaStrat, for adaptive stratification) is then:
\begin{equation}\label{eq:our_estimate}
	\Iest{AdaStrat} =
	\frac 1 N \sum_{p=1}^P \sum_{n=1}^{N_p} f(\tilde{X}_{n, p})
	\qquad \tilde{X}_{n, p} \sim \Unif(R_p),
\end{equation}
and it requires $2N$ evaluations of function $f$.

The whole procedure depends on a single tuning parameter, $d$, the maximum
depth of the tree. In practice, we set $d=k$, which means that $P=N$.  In a
standard use of decision trees, this would presumably lead to over-fitting, in
particular because rectangles obtained at a depth close to $k$ may contain too
few points to estimate properly the optimal splitting directions. 
Note however that, in our context, splitting a retangle (i.e. replacing one
stratum by two smaller strata) \emph{always} reduce the variance of the stratified
estimator.  We will return to his point in our convergence study, in
\cref{sec:convergence}.

\subsection{Generalisation to any $N$}\label{sub:gen}

One limitation of our approach is that it requires $N$ to be a power of two,
$N=2^k$.  To generalise it to an arbitrary $N$, we can adapt the first step as
follows.  We initialise again our collection of rectangles as
$\mathcal{R}=\left\{[0, 1]^s\right\}$. Then we split recursively the rectangles
in $\mathcal{R}$ in a way that ensures that the volume of each rectangle is
always of the form $n_R/N$, where $n_R\in \N^+$.

Specifically, given a rectangle $R$ of volume $n_R/N$, let
$R[n^+,j]^+$ and $R[n^-,j]^+$, be the two sub-rectangles obtained by splitting
$R$ along axis $j$, at cut-point $n^-/n_R$; the  volume of  $R[n^+, j]^+$
(resp. $R[n^-, j]^-$) is then $n^+/N$ (resp. $n^-/N$), with $n^+ + n^- = n_R$.
In practice, we choose $(n^+, j)$ so as to minimise:
\begin{multline*}
	\widehat{\Delta}(R, n^+, j) :=
	\frac{n^+}{n_R} \empvar\{Y_j: X_j\in R[n^+, j]^+\} \\
	+ \frac{(n_R - n^-)}{n_R}  \empvar\{Y_j: X_j\in R[n_R - n^+, j]^-\}.
\end{multline*}
This is very close to the standard variant of decision trees where one
optimises with respect to both the direction, $j$, and the cut-point.
However, in our case, we need to impose that the cut-point is $n^+/n_R$ for a
certain $n^+\in \N^+$. In this way, when the tree is fully grown, we obtain
a partition made of rectangles $R$, whose volume again is of the form $n_R/N$,
and, in the second step, we may simply generate $n_R$ uniform variates in each
rectangle $R$.

For the sake of simplicity, we focus on the $N=2^k$ version from now on.

\section{Convergence study}\label{sec:convergence}

Our objective in this section is to establish that the RMSE
of~\eqref{eq:our_estimate} converges at a rate faster than $\OO(N^{-1/2})$ for
certain classes of functions. We proceed in two steps: first, we obtain
convergence rates for a estimator based on an oracle tree; that is a tree which
would be constructed as in \cref{sub:tree}, but with empirical variances
replaced by true variances.   And, second, we study the impact of replacing the
oracle tree by an `estimated' one.

Before we start, we  discuss very briefly the state of the art regarding
the convergence of decision trees, and its relevance to our study. 

\subsection{Relation to the existing literature}\label{sub:review_conv}

Given the connection between stratification and control variates discussed in
the introduction, see \cref{sub:cv}, we thought initially of adapting 
existing results on the consistency (and rates of convergence) of decision
trees to obtain RMSE rates for our stratified estimators.  However, we found it
easier in the end to establish ours results directly, from first principles. 
We think this is due to two factors. 

First, establishing the consistency of approaches that rely on trees and greedy
algorithms is challenging. We refer in particular to the the review
of \cite{BiauScornetRandomForestTour}, who discusses the important gap between
theory and practice for such methods (including  random forests, which are a
generalisation of decision trees), and to \cite{ScornetBiauVertConsistency}, 
who establishes the  consistency of the corresponding estimates when the true
function is additive. In other words, existing results apply to small classes
of functions. 

Second, our settings are actually simpler in several ways that the ones
considered in the statistical literature: the true function $f$
is observed without noise in our case; we only allow binary splits in a given
direction (rather than splits at an arbitrary cut-point); and, while pruning is
essential in the standard use of decision trees (in order to avoid
over-fitting), it is unnecessary in our case. This is because splitting a
stratum into two strata \emph{always} reduce the variance of our estimator. 
The pruning step is one of the factors that make the study of decision tree
estimates challenging.

\subsection{Oracle trees}\label{oracle-trees}

We call oracle tree a tree obtained by the same procedure as in
\cref{sub:tree}, except that criterion $\widehat{\Delta}(R, j)$,
see~\eqref{eq:criterion_tree}, is replaced by its theoretical counterpart:
\[ \Delta(R, j) :=
	\frac 1 2 \Delta\left(R[j]^+\right)
	+ \frac 1 2 \Delta\left(R[j]^-\right),
	\quad \text{ with }
	\Delta(R) := \var\left[f(X) \middle| X \in R\right],
\]
and the depth of the tree is set to $k$, so that the number of rectangles in the final
partition is exactly $N=2^k$, and each rectangle has volume $1/N$. Let
\[
	\Iest{oracle} = \frac 1 N \sum_{n=1}^N f(\widetilde{X}_n),
	\quad \widetilde{X}_n \sim \Unif(R_n)
\]
be the corresponding estimator.

Of course, this estimator is not implementable in practice, but its convergence
rate gives us a lower bound on the error rate for the actual procedure. We
start with a basic result.

\begin{thm}\label{thm:linear}
	Assume $f(x) = \lambda^\top x$, with $\lambda\in\R^s$ such that
	$\|\lambda\|_0 = s_0 \leq s$. Then
	\[
		\rmse\left(\Iest{oracle}\right)
		= \OO(N^{-1/2-r(s_0)})
	\]
	with $r(s_0) = - \log(1 - 3/4s_0) / 2\log(2) \geq 0.541 / s_0$.
\end{thm}

This result strongly suggests that, while our approach may lead to higher
convergence rates, it is not able to achieve the best possible rate, $1/2+1/r$
(assuming $s_0=s$) on the class of $\setC{1}$ functions.

On the other hand, the fact that the rate depends on $s_0$, the `actual'
dimension of $f$ indicates that our procedure is sparsity adaptive. This is
hardly surprising, giving how the oracle tree is constructed: if
$\lambda[j]=0$, then $\Delta(R, j)=0$ for any rectangle $R$, so one never
splits along the $j$-th axis. Still, it is worth pointing out that, contrary to
Haber's estimator, or other methods based on a fixed stratification, our
approach will automatically detect irrelevant dimensions.

We now turn our attention to a more general class of functions.

\begin{thm}\label{thm:oracle_cls}
	Assume $f$ is strictly increasing (with respect to each component), and 
    that there exist real numbers $ 0 < \alpha_i < \beta_i$,
    $i=1,\ldots,s$, such that, for all $x,y\in [0,1]^s$,
	\[ 
	\Big(\sum_{i=1}^s\alpha_i^2(x_i-y_i)^2\Big)^{1/2}
    \leq \left|f(x)-f(y)\right|
    \leq \Big(\sum_{i=1}^s\beta_i^2(x_i-y_i)^2\Big)^{1/2}.
	\]
	Then we have
	\[ \rmse\left(\Iest{oracle}\right)
		= \OO(N^{-1/2-r(s)})
	\]
	with
	\[r(s) = \frac{1}{2\log 2}\log\left(\frac{\frac{3}{4} + \sum_{i=1}^s \frac{\beta^2_i}{\alpha^2_i}}{\sum_{i=1}^s \frac{\beta^2_i}{\alpha^2_i}}\right).
	\]
\end{thm}

Note that for large $s$ and $\beta_i \to \alpha_i$, we obtain essentially the
same rate as in \cref{thm:linear}. 

The strongest assumption in \cref{thm:oracle_cls} is the fact that $f$ must be
strictly increasing in every direction. One may wonder whether this condition
is really needed to obtain higher convergence rates. The following example
shows that this is indeed the case.

\begin{exa}
	Take $s=2$, and $f(x) = \lambda x[1] + \sin(2 \pi x[2])$, for some $\lambda >
		0$. Then it is easy to
	see that the oracle tree will never split along the second direction, because
	doing so would keep the variance unchanged. As a  result, the variance of
	$\Iest{oracle}$ cannot converge faster than $\OO(N^{-1/2})$.
\end{exa}

We note in passing that is is easy to extend \cref{thm:oracle_cls} to sparse
functions, that is functions that are constant in certain components $i$; for
such $i$, replace $\alpha_i$, $\beta_i$ and $\beta_i/\alpha_i$ by 0 in the
expressions in that theorem. In that sense, that result is also "sparsity
adaptive", as \cref{thm:linear}.

\subsection{Estimated trees}\label{estimated-trees}

To determine the behaviour of the actual estimator, we wish to study the
probability  that the  estimated tree (constructed as explained in
\cref{sub:tree}) matches the oracle tree. This raises two issues however.
First, if we split recursively the strata too many times, we end up with strata
that contain a small number of points. The quantities $\widehat{\Delta}(R, j)$,
used to decide along which axis one should split may become too noisy to ensure
that the right direction is chosen. 
This suggests studying instead the probability that the oracle tree and the
estimated tree matches only until a certain depth $L_k \ll k$. 

Second, even we do so, we may have cases where the choice of the direction is a
`close call', that is, there exists $j\neq j^\star$ such that $\Delta(R, j) -
\Delta(R, j^\star) \ll 1$. (Consider for instance the first splitting operation
when $f(x)=x[1] + 0.99 x[2]$.) In such a case, it is difficult to make the
probability of taking the right decision large.

To deal with this technical issue, we define (for any $\varepsilon>0$) the notion of an
$\varepsilon$-oracle tree, which is a tree constructed essentially as the oracle tree
(using the theoretical criterion $\Delta(R, j)$ to decide which direction $j$
to pick), except in cases when, while splitting rectangle $R$, there exists a
direction $j$ such that
\begin{equation}\label{eq:close_call}
	\Delta(R, j) \leq \Delta(R, j^\star) \frac{1+\varepsilon}{1 - \varepsilon}.
\end{equation}
In this situation, the $\varepsilon$-oracle is allowed to choose an arbitrary
direction $j$ (among those $j$ such that~\eqref{eq:close_call} holds). We are
able to show that the stratified estimator based on an $\varepsilon$-oracle
converges at a slightly slower rate than the one based on the oracle tree. 

With this approach, we are able to obtain convergence rates for our actual
estimator under essentially the same conditions as for the oracle estimator. 

\begin{thm}\label{thm:estimated_tree}
  Consider the set-up of \cref{thm:oracle_cls}. 
  % In addition, assume that
  % there constants $c\in(0,\infty)$ and $ \gamma\in(0,1)$ such that, for all
  % $k\in\mathbb{N}$, the tree has depth   $d_k$ is such that $d_k\leq  
  % k-c\log(k)^{1+\gamma}$, and assume that $\lim_{k\rightarrow\infty}d_k/k=1$.
  Then, one has \[ \rmse\left(\Iest{AdapStrat}\right) = \OO(N^{-1/2-r'}) \]
  for any $0<r'<r(s)$, where $r(s)$ is defined as in \cref{thm:oracle_cls}.
\end{thm}

% The technical condition on $d_k$ above seems to suggest we should  modify the first
% step of our procedure, and to grow the tree only a certain certain depth $d_k <
% k$ which fulfils the conditions of the theorem above. Recall however, as
% already discussed in \cref{sub:tree}, than splitting a rectangle always reduce
% the variance. Hence the theorem above also applies to any tree \emph{descended}
% from a tree with depth $d_k$, including a tree with depth $k$ as described in
% \cref{sub:algo}.

\section{Numerical experiments}\label{sec:numerics}

In this section, we compare our adaptive stratification estimator with the
standard Monte Carlo estimator, $\Iest{MC}$, and Haber estimator or order one,
$\Iest{Haber}$ when possible (recall that the last estimator is defined only
for $N=k^s$, $k\geq 2$).  Our comparison is in terms of relative RMSE (RMSE
divided by true value, or an estimate based on a grand mean), and how it 
varies as a function of $N$. The RMSE is evaluated over 256 independent
realisations of the estimators. 

\subsection{Toy example}

We consider the following linear function $f(x) = \exp\{\sum_{i=1}^s\lambda_i
x_i\}$, with decreasing weights,   $\lambda_i = i^{-2}$.  This way, component
$i$ contributes less and less to the overall variance as $i$ grows; this is a
favourable case for our adaptive stratification strategy.

\cref{fig:toy_rmse} compares the performance (i.e. RMSE vs $N$) of the following
unbiased estimators when $s=5$: standard Monte Carlo, $\Iest{MC}$, Haber of order one
(defined in \cref{sec:Introduction}), and our stratified estimator based on a
decision tree grown either until depth $d=k$.
Interestingly, our approach outperforms even Haber's estimator, although the
latter is supposed to converge at a higher rate. The dashed lines represents
the theoretical convergence rates of the considered estimators. 

\cref{fig:toy_rmse_higher_dim} does the same comparison for $s=15$ and $s=50$, but this
time Haber's estimate is omitted, as it is only defined for $N=k^s$, $k\geq 2$,
and thus it cannot be computed for reasonable values of $N$ whenever $s \gg
10$. We note that our approach leads to a measurable improvement even when
$s=50$. 

\begin{figure}
  \centering \includegraphics[width=0.5\linewidth]{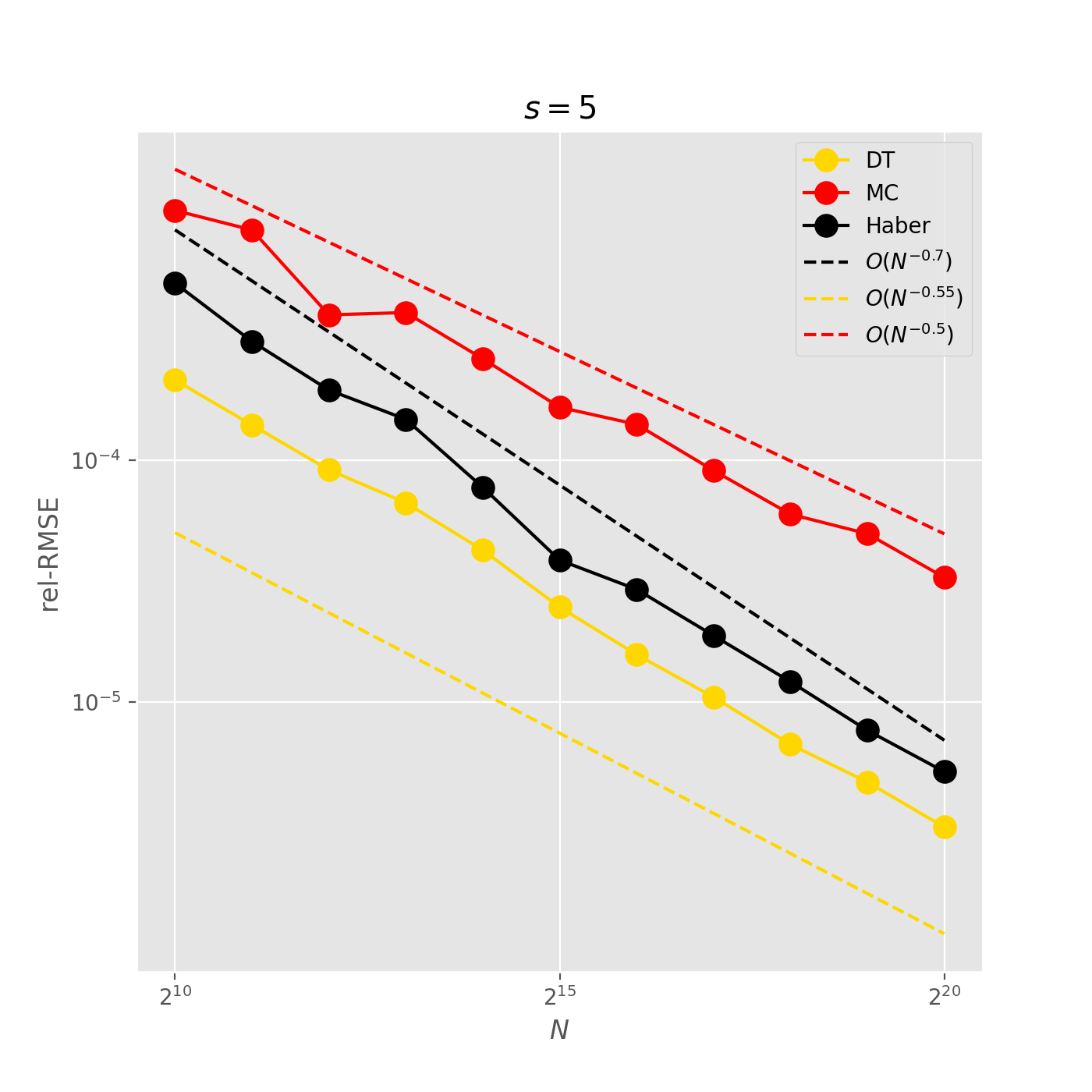}
  \caption{Toy example, $s=5$. RMSE vs $N$ (estimated from 256 independent
    realisations) for the following unbiased estimators: standard Monte Carlo
    (MC, red), Haber of order one (Haber, black), and our adaptive decision-tree strategy
    (DT, yellow). Dashed lines represent the theoretical convergence rates of
  these three estimators.}
\label{fig:toy_rmse}
\end{figure}

\begin{figure}
  \centering  
  \includegraphics[width=0.39\textwidth]{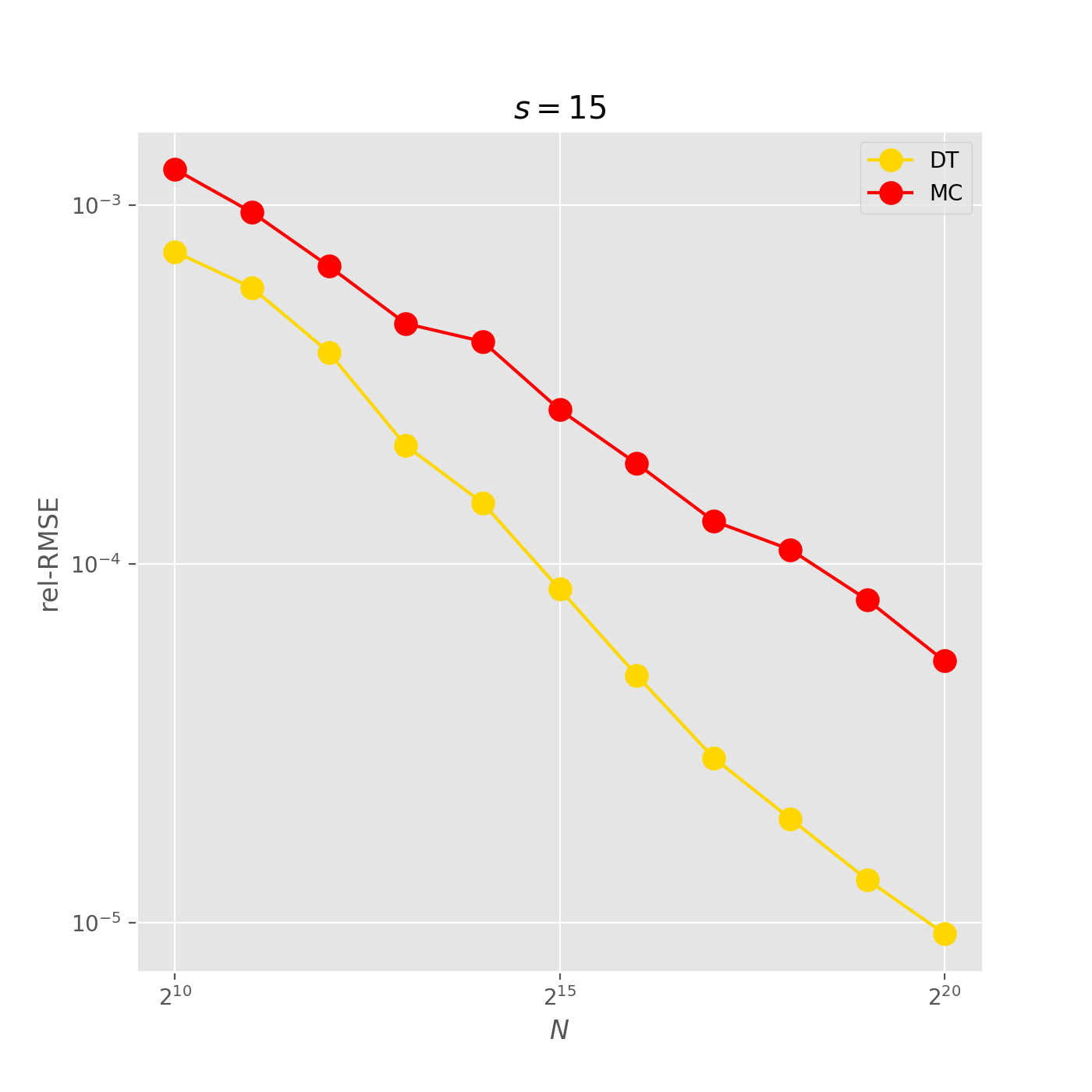}
  \includegraphics[width=0.39\textwidth]{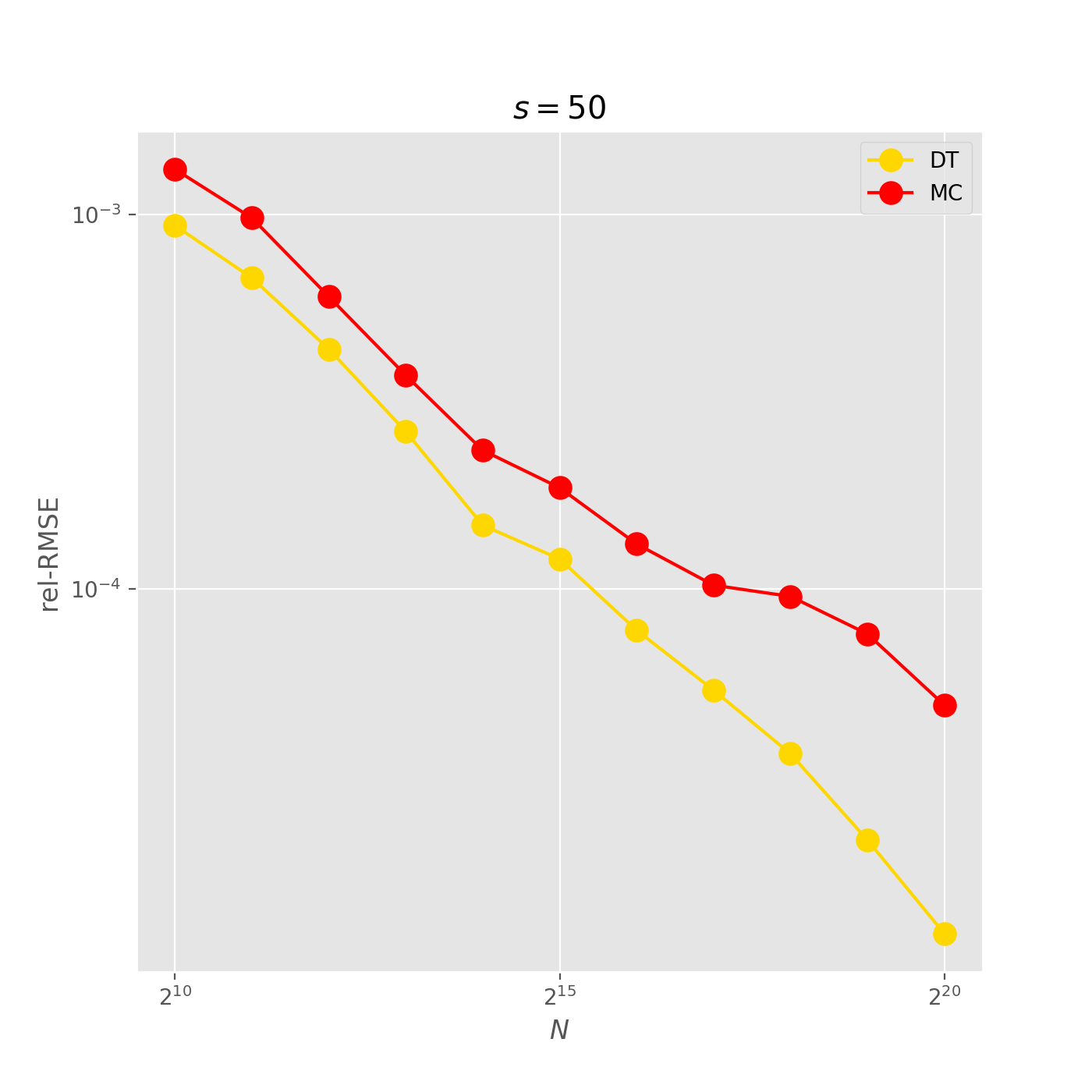}
  \caption{Toy example, $s=15$ (left) and $s=50$ (right). Same plot as
    \cref{fig:toy_rmse} except Haber estimator of order one has been excluded,
    and the theoretical convergence rates are not plotted for
    the sake of readability. }
    \label{fig:toy_rmse_higher_dim}
\end{figure}

One advantage of stratification that we have not discussed in the body of the
paper is that  it is easy to estimate the variance of a stratified estimator,
by generating two (instead of one) point in each stratum. The variance
estimate is the sum (over the strata) of the empirical variance within each
stratum. This idea was discussed in relation to Haber and related estimators in
\cite{strat}.  \Cref{fig:toy_boxplots_var} compares box-plots of such variance
estimates when $N=2^{10}$ and $s=5$ with the variance estimate of standard
Monte Carlo. One can see that one obtains variance estimates that are far more
stable. 

\begin{figure}[ht] \centering 
  \includegraphics[width=0.6\textwidth]{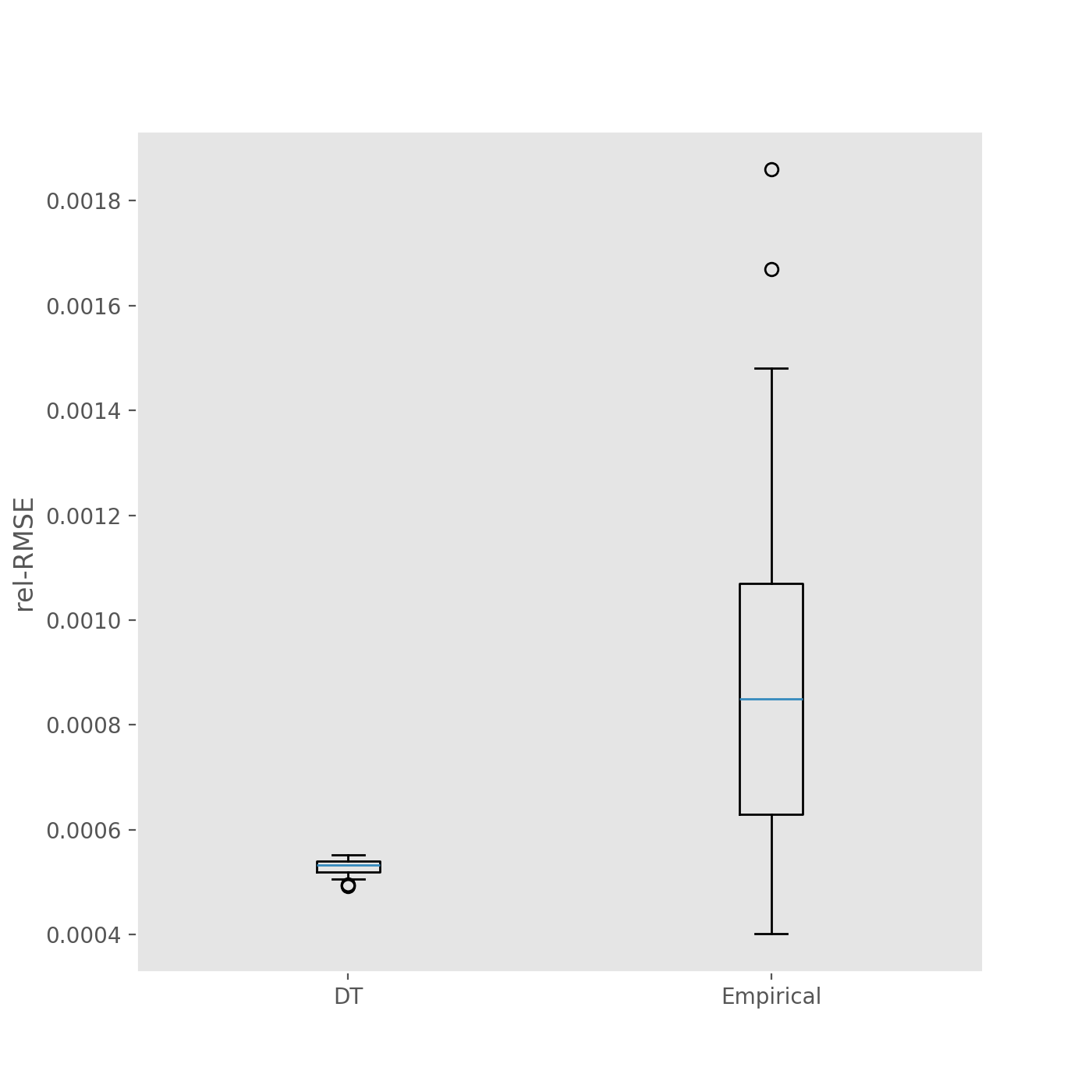} 
  \caption{Toy example, $s=5$, $N=2^{10}$. Box-plots of variance estimates for
  the Monte Carlo estimator and our adaptive stratified estimator.}
  \label{fig:toy_boxplots_var}
\end{figure}

\subsection{Bayesian model choice}

We consider a Bayesian statistical model, with parameter $\beta\in \R^d$, prior
distribution $p(\beta)$, and likelihood $L(y | \beta)$. We wish to approximate
unbiasedly  the marginal likelihood $p(y) = \int p(\beta)L(y|\beta)\,d\beta$.
This quantity may be used to perform model choice.  We adapt the importance
sampling approach described in \cite{mr3634307} to  approximate marginal
likelihoods as follows: we obtain numerically the mode $\hat\beta$, and the
Hessian at $\beta=\hat\beta$, of the function $h(\beta)=\log\{
p(\beta)L(y|\beta)\}$; hence  $h(\beta)\approx h(\hat\beta)- (1/2)
(\beta-\hat\beta)^T H(\beta -\hat\beta)$. Then we set the importance density to
$q(\beta) = N(h(\hat\beta),H^{-1})$. Finally, we rewrite the integral 
\[ p(y) = 
  \mathbb{E}_{q}\left[\frac{p(\beta)L(y|\beta)}{q(\beta)}\right] 
  = \int_{[0, 1]^s} f(x) \dd{x}
\]  
where $f(x) = \hat{\beta} + C \Phi^{-1}(x)$, $C C^\top = H^{-1}$ is the
Cholesky decomposition of $H^{-1}$, and $\Phi^{-1}(x)$ is the vector  of
$\Phi^{-1}(x_i)$'s, with $\Phi^{-1}$ being the  the quantile (inverse CDF)
function associated to the $N(0, 1)$ distribution. 

Specifically, we assume a Gaussian prior (with zero mean and covariance $5^2
I_s$), and a logistic regression model:  $L(y|\beta) = \prod_{i=1}^n F(y_i
\beta^T x_i)$, $F(z)=1/(1+e^{-z})$, and the data $(x_i, y_i)_{i=1}^{n}$ consist
of predictors $x_i\in \R^s$ and labels $y_i\in\{-1, 1\}$.  We use the
German credit dataset from the UCI machine learning repository, and the $s$
first predictors, for $s=5$, $10$, $15$, $20$.
Predictors are rescaled to have mean zero and variance one.

Figure \ref{fig:german} compares our estimators with the standard MC (Monte
Carlo) estimator for $s=5$, 10, 15 and 20. Whatever the dimension, we see a
distinctive improvement with our method. 

\begin{figure}
\centering
\includegraphics[width=0.39\textwidth]{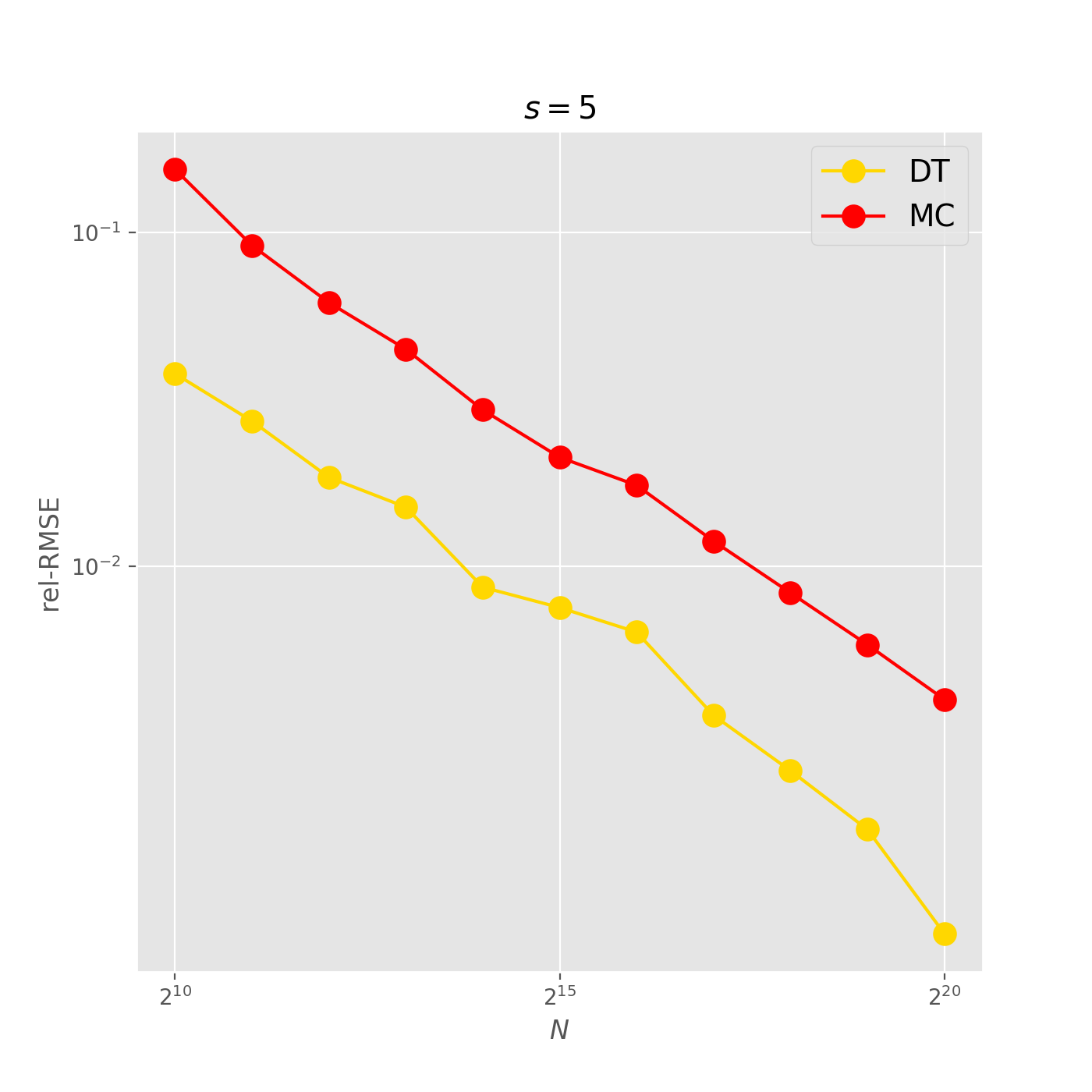}
\includegraphics[width=0.39\textwidth]{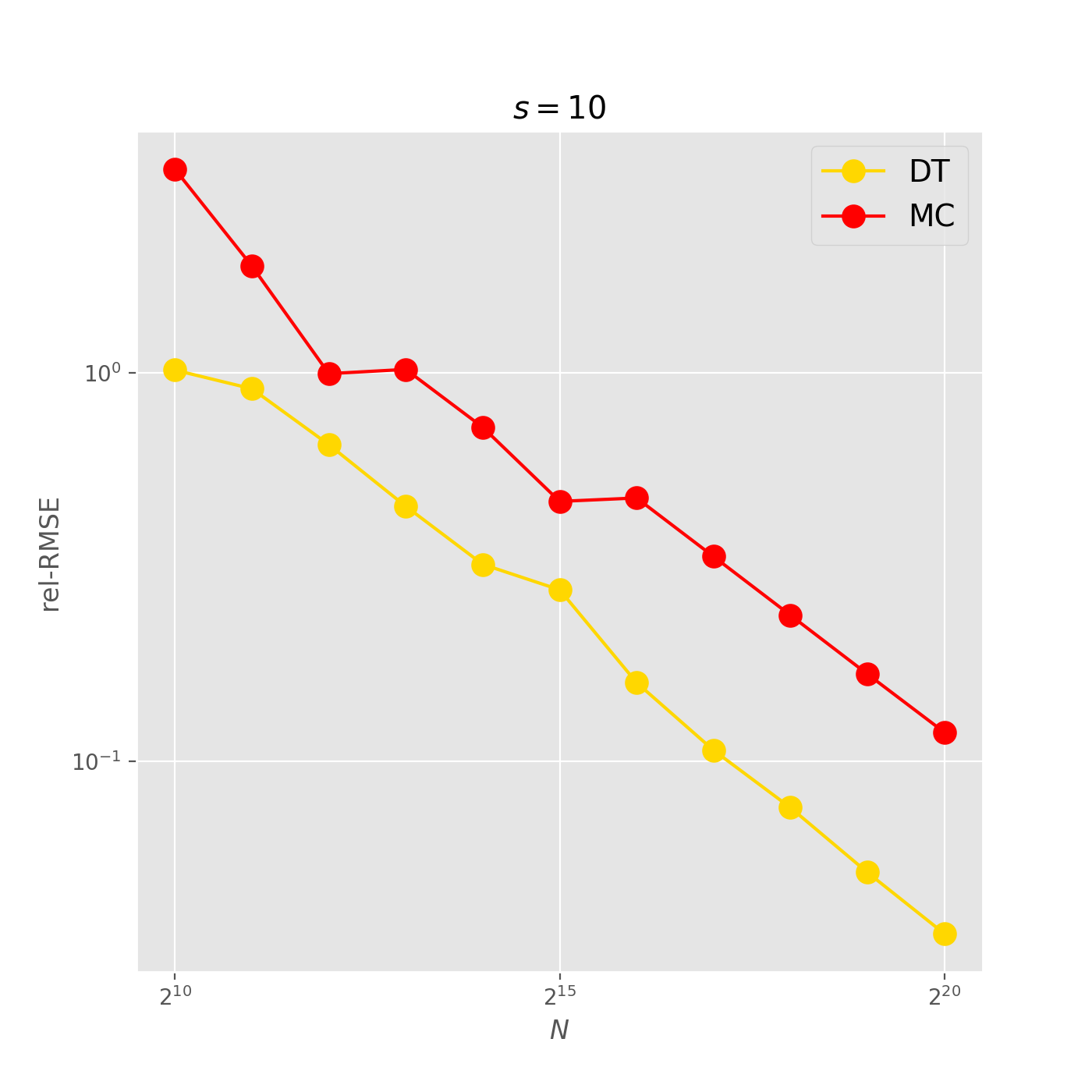}
\includegraphics[width=0.39\textwidth]{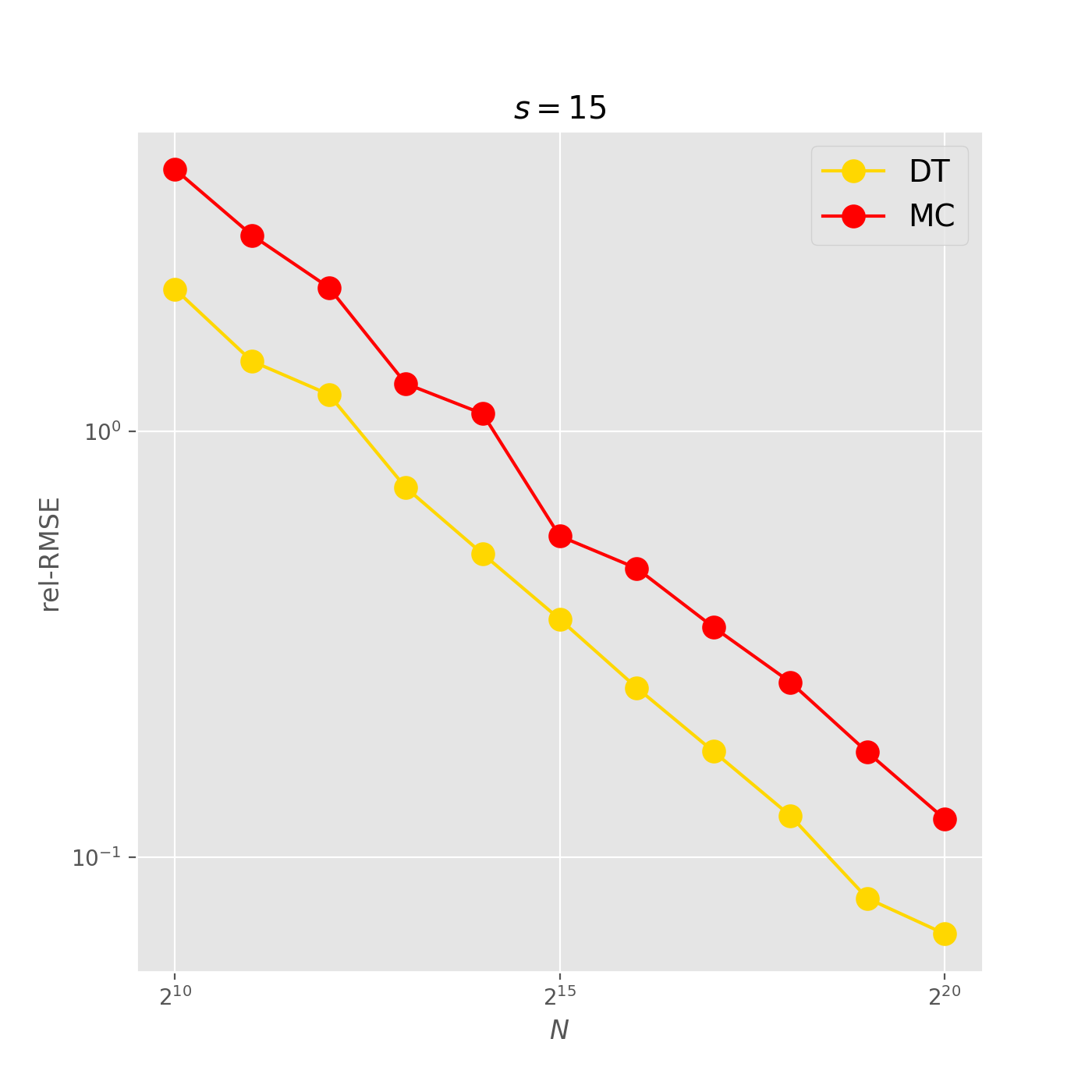}
\includegraphics[width=0.39\textwidth]{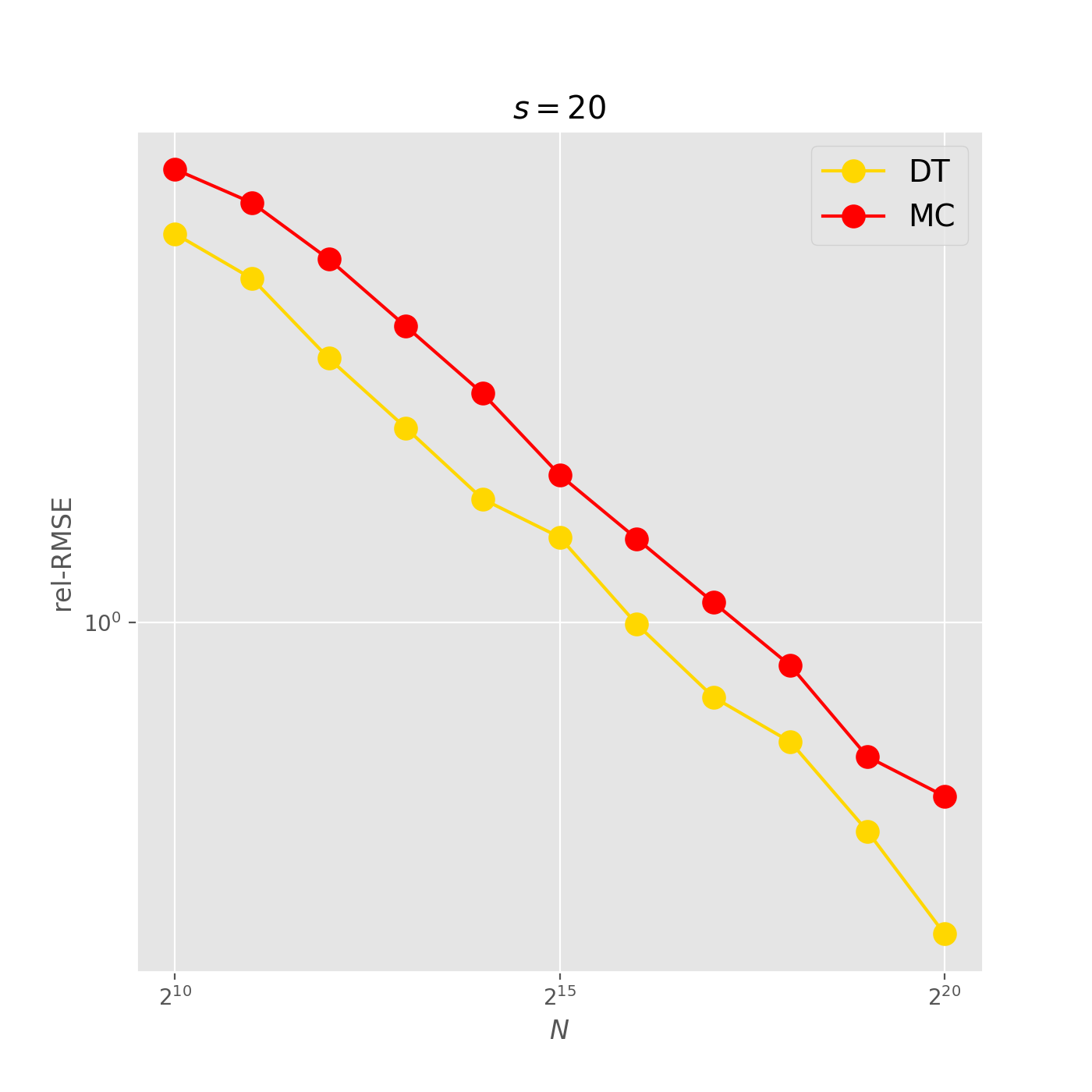}
\caption{Bayesian model choice example, $s=5$, $10$, $15$, $20$: RMSE  (over
  256 independent realisations) versus $N$ for the following unbiased
  estimators: standard Monte Carlo (MC, red), and our adaptive stratified Monte
Carlo estimators (DT, yellow).}
 \label{fig:german}
\end{figure}

\section{Discussion and future work}\label{sec:conclusion}

The strategy developed in this paper represents an important step towards
making stratified Monte Carlo more practical, especially in moderate to high
dimensions. Its main limitation is that it does  not reach the optimal
convergence rate for $\mathcal{C}^r$ functions, even for $r=1$. It remains to
be seen whether one can obtain such rates with a tree-based approach.  More
generally, we suspect that the decision trees and their
generalisations such as random forest may find other uses in Monte Carlo
computation (e.g., control variates), in particular because of the low
complexity of their learning procedures.

\section*{Acknowledgements}\label{sec:Acknowledgements}

The first author is partly supported by ANR project BLISS.
The third author wishes to thank CREST for supporting her visit to the first
author during her PhD.

\bibliographystyle{apalike}
\bibliography{complete}

\appendix

\section*{Proofs}\label{sec:proof}

\subsection*{Proof of \cref{thm:linear}}

We consider a sequence of stratified estimators, $\Ih_0, \Ih_1, \dots$, which
corresponds to a sequence of sets of rectangles $\mathcal{R}_0, \mathcal{R}_1,
	\dots$, starting at $\mathcal{R}_0=\{[0, 1]^s\}$ (hence $\Ih_0=\Iest{MC}$)
and corresponding to the successive splitting operations performed while
growing the tree, as described in \cref{sub:tree}: that is, at each iteration,
takes a rectangle $R$ in $\mathcal{R}_t$ and replace it by two rectangles,
$R[j^\star]^+$ and  $R[j^\star]^-$, where $j^\star$ is the index that leads to
greatest variance reduction.

These stratified estimators $\Ih_0$, $\Ih_1$, $\dots$ are defined as
\[
	\Ih_t = \frac 1 N \sum_{R\in\mathcal{R}_t} \sum_{n=1}^{n_R} f(X_{p,n}),
	\qquad X_{p, n} \sim \Unif(R_p),
\]
where $n_R=N \times \vol(R)$, hence their variance has the form:
\[
	\frac{1}{N} \sum_{R\in\mathcal{R}_t} \vol(R) \times
	\Delta(R)
\]
where $\Delta(R) = \var\left[f(X) \middle| X \in R\right]$, where the
conditional distribution is taken with respect to $X\sim\Unif([0, 1])$.
Thus, splitting a particular rectangle amounts to replacing one term in this
expression by a sum of two terms, leading to the variance reduction factor:
\begin{equation}\label{eq:var_reduction}
	\frac{\Delta(R[j^\star]^+) + \Delta(R [j^\star]^-)}{2\Delta(R)}.
\end{equation}

Elementary calculations show that, if $f(x)=\lambda^T x$, and rectangle $R$ has
dimensions $\mu_1\times\dots\times \mu_s$, then $\Delta(R) = \sum_{j=1}^s
	\mu_j^2\lambda_j^2/12$. In that case, ratio~\eqref{eq:var_reduction} equals:
\[
	\frac{\Delta(R[j^\star]^+) + \Delta(R [j^\star]^-)}{2\Delta(R)} =
	1 - \frac{3 \mu_{j^\star}^2\lambda_{j^\star}^2}{4 \sum_{j=1}^s \mu_j^2\lambda_j^2}
	\leq 1 - \frac{3}{4s_0}.
\]
The inequality comes from $j^\star = \arg\min_{j=1,\dots,s}
	\{\Delta(R[j]^+) + \Delta(R[j]^-)\}$, and the fact that the maximum of $s$
number is larger than their average.

At the end of the splitting process, where one has obtained $N$ rectangles of
volume $1/N$ (recall that $d=k$ hence $P=2^d=2^k=N$), we get
\[
	\var[\Iest{oracle}]  \leq \left(1 - \frac{3}{4s_0}\right)^k \var[\Ih_0]
\]
with $\var[\Ih_0] = \OO(N^{-1})$, and, since $N=2^k$,
\[
	\left(1 - \frac{3}{4s_0}\right)^k =
	\exp\left\{ \log(N) \times \frac{\log(1-3/4s_0)}{\log 2} \right\}
	= N^{-2\alpha(s_0)}
\]
with $\alpha(s_0) = -\log(1 - 3/4s_0)/(2\log 2) \geq 3/(8 s_0 \log 2)$,
since $\log(1 + x) \leq x$.

\subsection*{Proof of \cref{thm:oracle_cls}}

We follow the same lines as in the previous proof: given a rectangle $R$, we
wish to bound variance reduction factor~\eqref{eq:var_reduction}. We use the
law of total variance formula to get (for any $j\in 1:s$):
\begin{equation}
	\Delta(R) = \var \left[ f(X) \middle| X \in R \right] = C_1 + C_2
\end{equation}
with
\begin{align*}
	C_1 & := \frac 1 4 \Big\{ \E\left[f(X) \middle| X \in R[j]^+\right] -
	\E\left[f(X) \middle| X \in R[j]^-\right] \Big\}^2,                         \\
	C_2 & := \frac 1 2 \var\left[f(X) \middle| X \in R[j]^+\right]
	+  \frac 1 2 \var\left[f(X) \middle| X \in R[j]^-\right].
\end{align*}

To bound the first term, we remark that (assuming the dimensions of $R$ are
$\mu_1\times \dots\times \mu_s$):
\begin{multline*}
	\E\left[f(X) \middle| X \in R[j]^+\right]
	- \E\left[f(X) \middle| X \in R[j]^-\right] \\
	= \E\left[\left\{f(X+\frac{\mu_{j}}{2} e_{j}) - f(X) \right\} \middle| X \in R[j]^-\right]
	\geq \frac{\alpha_i \mu_{j}}{2}
\end{multline*}
hence $C_1 \geq \alpha_j^2 \mu_j^2 / 16$.

To bound the second term, we remark that:
\begin{align*}
	\var\left[f(X) \middle| X \in R \right]
	 & = \frac{1}{2} \E\left[ \left\{ f(X') - f(X) \right\}^2 \middle| X, X' \in R \right]
	\leq \frac{1}{12} \sum_{i=1}^s \beta_i^2 \mu_i^2. 
\end{align*}
If we apply this inequality to $R[j]^+$ and $R[j]^-$, we obtain:
\begin{align*}
	C_2 &\leq    \frac{1}{12}  \sum_{i=1}^s \beta^2_i\mu_i^2 =\frac{\alpha^2_j\mu^2_j}{12}  \sum_{i=1}^s \frac{\beta^2_i}{\alpha^2_i}\,\frac{\alpha_i^2\mu_i^2}{\alpha^2_j\mu^2_j}.
\end{align*}

Then, we can bound~\eqref{eq:var_reduction} as follows:
\begin{align*}
	\frac{\Delta(R[j]^+) + \Delta(R [j]^-)}{2\Delta(R)}
	= \frac{C_2}{C_1 + C_2}
	 &\leq \frac{\frac{\alpha^2_j\mu^2_j}{12}  \sum_{i=1}^s \frac{\beta^2_i}{\alpha^2_i}\,\frac{\alpha_i^2 \mu_i^2}{\alpha^2_j\mu^2_j}}{\frac{\alpha_j^2 \mu_j^2}{ 16}+\frac{\alpha^2_j\mu^2_j}{12}  \sum_{i=1}^s \frac{\beta^2_i}{\alpha^2_i}\,\frac{\alpha_i^2 \mu_i^2}{\alpha^2_j\mu^2_j}}=\frac{\frac{4}{3}  \sum_{i=1}^s \frac{\beta^2_i}{\alpha^2_i}\,\frac{\alpha_i^2 \mu_i^2}{\alpha^2_j\mu^2_j }}{1 +\frac{4}{3}  \sum_{i=1}^s \frac{\beta^2_i}{\alpha^2_i}\,\frac{\alpha_i^2 \mu_i^2}{\alpha^2_j\mu^2_j}}.
\end{align*}
This inequality holds for all $j=1,\dots, s$, and in particular it holds for
$j^\star = \arg\min\{ \Delta(R[j]^+) + \Delta(R[j]^-)$. Therefore, noting that if
we let $k=\mathrm{argmax}_{j\in\{1,\dots,s\}}\alpha_j^2\mu_j^2$ we have
\begin{align*}
\frac{\frac{4}{3}  \sum_{i=1}^s \frac{\beta^2_i}{\alpha^2_i}\,\frac{\alpha_i^2 \mu_i^2}{\alpha^2_j\mu^2_j }}{1 +\frac{4}{3}  \sum_{i=1}^s \frac{\beta^2_i}{\alpha^2_i}\,\frac{\alpha_i^2 \mu_i^2}{\alpha^2_j\mu^2_j}}\geq \frac{\frac{4}{3}  \sum_{i=1}^s \frac{\beta^2_i}{\alpha^2_i}\,\frac{\alpha_i^2 \mu_i^2}{\alpha^2_{k}\mu^2_{k} }}{1 +\frac{4}{3}  \sum_{i=1}^s \frac{\beta^2_i}{\alpha^2_i}\,\frac{\alpha_i^2 \mu_i^2}{\alpha^2_{k}\mu^2_{k}}},\quad\forall j\in\{1,\dots,s\}
\end{align*}
it follows that 
\begin{align*}
\frac{\Delta(R[j^\star]^+) + \Delta(R [j^\star]^-)}{2\Delta(R)}&\leq  \frac{\frac{4}{3}  \sum_{i=1}^s \frac{\beta^2_i}{\alpha^2_i}\,\frac{\alpha_i^2 \mu_i^2}{\alpha^2_{k}\mu^2_{k} }}{1 +\frac{4}{3}  \sum_{i=1}^s \frac{\beta^2_i}{\alpha^2_i}\,\frac{\alpha_i^2 \mu_i^2}{\alpha^2_{k}\mu^2_{k}}}\leq \frac{\frac{4}{3}  \sum_{i=1}^s \frac{\beta^2_i}{\alpha^2_i}}{1 +\frac{4}{3}  \sum_{i=1}^s \frac{\beta^2_i}{\alpha^2_i}}
\end{align*}
where the second inequality uses the fact that $(\alpha_i^2\mu_i^2)/(\alpha_{k}^2\beta_{k}^2)\leq 1$ for all $i\in\{1,\dots,s\}$ and the fact that the function $x\mapsto x/(1+x)$ is non-decreasing on $[0,\infty)$.

We can  now conclude in the same way as in the previous proof:
\begin{align}
	\var\left[\Iest{oracle}\right]
	 & \leq \left(\frac{ \sum_{i=1}^s \frac{\beta^2_i}{\alpha^2_i}}{\frac{3}{4} + \sum_{i=1}^s \frac{\beta^2_i}{\alpha^2_i}} \right)^k
     \var\left[\Iest{MC}\right] \label{eq:rate_oracle}\\
	 & = \OO(N^{-1-2r(s)}) \nonumber
\end{align}
with 
\begin{align*}
r(s)=\frac{1}{2\log 2}\log\left(\frac{\frac{3}{4} + \sum_{i=1}^s \frac{\beta^2_i}{\alpha^2_i}}{\sum_{i=1}^s \frac{\beta^2_i}{\alpha^2_i}}\right).
\end{align*}

\subsection*{Proof of \cref{thm:estimated_tree}}

We consider some fixed $\varepsilon\in(0, 1)$ throughout.

We first recall a concentration inequality due to
\cite{MaurerPontil2009Empirical}.
\begin{prop}[Theorem 10 in \citealp{MaurerPontil2009Empirical}]
	Let $Z_1,\ldots, Z_N$ be IID random variables taking values in $[0, 1]$,
	$\Delta=\var(Z)$, and
	\[ \widehat{\Delta} = \frac{1}{N-1}\sum_{n=1}^N (Z_n - \bar{Z})^2,
	\]
	the empirical variance. Then:
	\begin{align}
		\pr\left( \widehat{\Delta}^{1/2} > \Delta^{1/2} + \varepsilon \right)
		 & \leq \exp\left\{ - \frac{(N-1)\varepsilon^2}{2} \right\}, \label{eq:pontil_geq} \\
		\pr\left( \widehat{\Delta}^{1/2} < \Delta^{1/2} - \varepsilon \right)
		 & \leq \exp\left\{ - \frac{(N-1)\varepsilon^2}{2} \right\}. \nonumber
	\end{align}
\end{prop}

A simple corollary of the above result is:
\begin{lem}\label{lem:ratio_var}
	Let $Z_1,\dots,Z_N$ IID variables taking values in some compact interval $[a,
				b]$, $\Delta$, $\widehat{\Delta}$ be defined as above, and
	$\gamma:= (1 + \varepsilon)^{1/2} - 1$.
	Then (assuming $\Delta>0$):
	\begin{align*}
		\pr\left( \frac{\widehat{\Delta}}{\Delta} > 1 + \varepsilon \right)
		 & \leq \exp\left\{ - \frac{(N-1)\Delta\gamma^2}{2(b-a)^2} \right\},  \\
		\pr\left( \frac{\widehat{\Delta}}{\Delta} < 1 - \varepsilon \right)
		 & \leq \exp\left\{ - \frac{(N-1)\Delta\gamma^2}{2 (b-a)^2} \right\}.
	\end{align*}
\end{lem}

\begin{proof}
	Inequality $\widehat{\Delta} > \Delta (1 + \varepsilon)$ is equivalent to
	\[
		\widehat{\Delta}^{1/2} > \Delta^{1/2} + \Delta^{1/2 } \{ (1 +
		\varepsilon)^{1/2} - 1 \}
	\]
	the probability of which may be bounded by applying \cref{eq:pontil_geq} to
	the rescaled variables $Z_n' = (Z_n - a) / (b - a)$, and replacing
	$\varepsilon$ with $\Delta^{1/2} \gamma$, with $\gamma=(1+\varepsilon)^{1/2}
		- 1$.
\end{proof}

We can now apply this lemma to the set of $f(X_n)$ for those $X_n$ that fall in
a given rectangle $R$.

\begin{lem}\label{lem:rectangle_bound}
	Let $R$ a rectangle included in $[0, 1]^s$, of volume $2^{-l}$. Then, under the same conditions as in \cref{thm:oracle_cls},
	\begin{equation*}
		\pr\left( \frac{\widehat{\Delta}(R)}{\Delta(R)} \geq 1 + \varepsilon \right)
		\leq e^\kappa
		\left[1 - 2^{-l} (1- e^{-\kappa}) \right]^N
	\end{equation*}
	where $\kappa :=\gamma^2\alpha^2/12\beta^2$, $\gamma= (1 + \varepsilon)^{1/2} - 1$,  $\alpha=\inf_{i\in\{1,\dots,s\}}\alpha_i$ and $\beta\max_{i\in\{1,\dots,s\}}\beta_u$. This inequality also holds for
	$\pr\left( \frac{\widehat{\Delta}(R)}{\Delta(R)} \leq 1 - \varepsilon
		\right)$.
\end{lem}

\begin{proof}
	Let $N(R)=\sum_{n=1}^N \ind\{X_n\in R\}$. Conditional on $N(R)=n$, the $n$
	variates $X_n$ that fall in $R$ follow an uniform distribution with respect to
	$R$. Therefore, from \cref{lem:ratio_var}, one has:
	\begin{align*}
		\pr\left( \frac{\widehat{\Delta}(R)}{\Delta(R)} \geq 1 + \varepsilon
		\middle| N(R) = n \right)
		 & \leq \exp\left\{ - \frac{(n-1)\Delta(R)\gamma^2}{2D(R)} \right\}    \\
		 & \leq \exp\left\{ - (n-1) \frac{\gamma^2\alpha^2}{12\beta^2}\right\}
		\\
		 & \leq e^\kappa \exp\left\{ - n \kappa \right\}
	\end{align*}
	where in the first line, $D(R) = \left(\sup_{x,y\in
			R}|f(y)-f(x)|\right)^2$, and, in the second line, we have used
	\begin{equation}
		\Delta(R) \geq \frac{\alpha^2}{2} \E\left[\|X - Y\|^2 \middle| X, Y\in R\right]
		= \frac{\alpha^2}{12} \sum_{i=1}^s \mu_i^2
	\end{equation}
	and $D(R) \leq \beta^2 \sum_{i=1}^s \mu_i^2$, assuming that $R$ has
	dimensions $\mu_1\times \dots\times \mu_s$.
	Since $N(R)\sim\mathrm{Bin}(N, 2^{-l})$, one has, for any $c>0$:
	$\E[\exp\{-cN(R)\}] = [1 -  2^{-l} + 2^{-l} \exp(-c)]^N$, which gives the
	desired result.
	The second inequality may be established in the same way.
\end{proof}

We can now proceed with the proof of the theorem.
The gist of the proof is to show that the probability that the estimated tree
is an $\varepsilon$-oracle is close to one. Consider a rectangle $R$, of volume
$2^{-l}$ (i.e., it has been obtained through $l$ splitting operations), and
assume that~\eqref{eq:close_call} does not hold for any $j$. Then, one chooses the
right coordinate, $j^\star = \arg\min_{j=1,\dots, s}\Delta(R, j)$ as soon as:
for all $j\neq j^\star$, $\widehat{\Delta}(R[j]^+) \geq \Delta(R[j]^+) (1 -
	\varepsilon)$, the same is true for $R[j]^-$, and, in addition,
$\widehat{\Delta}(R[j^\star]^+) \leq \Delta(R[j^\star]^+) (1 + \varepsilon)$,
and the same holds for $R[j^\star]^-$.

Conversely, using the union bound and \cref{lem:rectangle_bound}, we see that
the probability of not choosing the right coordinate (when splitting $R$) is
bounded by 
\[ 2 s e^\kappa \left(1 - 2^{-l-1} e^{-\kappa} \right)^N  \]
with the constant $\kappa>0$  as defined in  \cref{lem:rectangle_bound}.

More generally, the probability $p_{k,L}$ of making at least one `wrong' decision  while
constructing the tree (i.e. constructing a tree that is an
$\varepsilon$-oracle), until depth $L$, is such that
\begin{equation}\label{eq:bound_pl}
\begin{split}
	p_{k,L}\leq 2 s e^\kappa \sum_{l=1}^L 2^l \left(1 - 2^{-l} e^{-\kappa} \right)^N
	 & \leq 2 s e^\kappa L 2^L \left(1 - 2^{-L} e^{-\kappa} \right)^N                  \\
	 & \leq 2 s e^\kappa \exp\left\{ \log L + L \log 2 - 2^{k-L} e^{-\kappa}  \right\}.
\end{split}
\end{equation}
Remark that for all $\delta\in(0,\infty)$ we have
\begin{align*}
\exp\big(- \delta 2^{k-L}  \big)\leq   2^{-2 k}\Leftrightarrow L\leq
k-\frac{\log(2 k/\delta)}{  \log(2)}-\frac{\log(\log(2))}{\log(2)}. 
\end{align*}

Take $L_k = \lfloor k - c \log(k)^{1+\gamma} \rfloor $ for some constants
$c\in(0,\infty)$ and $\gamma\in(0,1)$. Then 
\begin{align}\label{eq:bound_p2} \exp\big(-  2^{k-L_k}e^{-\kappa}   \big)\leq  
  2^{-2k}=N^{-2}.
\end{align}
In words, the probability that the estimated tree and the $\varepsilon-$oracle
tree differ before depth $L_k$ is $\OO(N^{-2})$. Beyond depth $L_k$, we let the
$\varepsilon-$oracle tree choose arbitrarily the split decisions; this leads to
variance reduction factors which equal 1 in the worst case (splitting always
reduce the variance).

We conclude the proof by remarking that replacing the oracle by an
$\varepsilon$-oracle  degrades the convergence rate by a certain amount. More
precisely, adapting~\eqref{eq:rate_oracle} leads to 
\begin{align*}
  \var\left[\Ih_{\varepsilon-\mathrm{oracle}}\right] & \leq \left(\frac{
      \sum_{i=1}^s \frac{\beta^2_i}{\alpha^2_i}}{\frac{3}{4} + \sum_{i=1}^s
    \frac{\beta^2_i}{\alpha^2_i}}  \times \frac{1 + \varepsilon}{1 -
\varepsilon}\right)^{L_k}\\
     &=\left\{\left(\frac{ \sum_{i=1}^s
           \frac{\beta^2_i}{\alpha^2_i}}{\frac{3}{4}  + \sum_{i=1}^s
         \frac{\beta^2_i}{\alpha^2_i}}  \times \frac{1 + \varepsilon}{1 -
   \varepsilon}\right)^{1-(1-L_k/k)}\right\}^{k}\\
     &\leq \left\{\left(\frac{ \sum_{i=1}^s
           \frac{\beta^2_i}{\alpha^2_i}}{\frac{3}{4}  + \sum_{i=1}^s
         \frac{\beta^2_i}{\alpha^2_i}}  \times \frac{1 + \varepsilon}{1 -
   \varepsilon}\right)^{1-\varepsilon}\right\}^{k}
    \end{align*}
    where the  last inequality holds for $k$ large enough. This shows that 
    $\var[\Iest{MC}]= \OO(N^{-1-2r'_\varepsilon})$ with
    \begin{align*}
      r_\varepsilon'=-\frac{1-\varepsilon}{2\log 2}\log\left(\frac{
          \sum_{i=1}^s \frac{\beta^2_i}{\alpha^2_i}}{\frac{3}{4} + \sum_{i=1}^s
        \frac{\beta^2_i}{\alpha^2_i}}  \times \frac{1 + \varepsilon}{1 -
    \varepsilon}\right)<r(s).
  \end{align*}
  Since $\varepsilon\in(0,1)$ is arbitrary one can make $r_\varepsilon'$ as
  close as $r(s)$ as desired.

\end{document}